\documentclass[11pt]{memoir}
\setlrmargins{*}{*}{1}
\checkandfixthelayout

\counterwithout{section}{chapter}
\counterwithout{figure}{chapter}

\firmlists

\usepackage{amssymb}

\usepackage{enumitem}
\setlist[enumerate]{leftmargin=*}
\setlist[itemize]{leftmargin=*}
\setlist[description]{font=\mdseries\textsf, leftmargin=1.5em}
\usepackage{hyperref}
\usepackage{memhfixc} 
\usepackage[numbered]{bookmark} 

\usepackage[all]{hypcap} 

\usepackage{float}

\usepackage[charter]{gacs} 
\newcommand\CA {\mathrm{CA}}
\newcommand\trans {\mathrm{Trans}}

\newcommand\Forks{\mathrm{Forks}}
\newcommand\Noise{\mathrm{Noise}}
\newcommand\Size{\mathrm{Size}}
\newcommand\Span{\mathrm{Span}}
\newcommand\Excuse{\mathrm{Excuse}}
\newcommand\Expl{\mathrm{Expl}}
\newcommand\Time {\mathrm{Time}}

\newcommand\proj {\mathrm{proj}}
\renewcommand{\P}{\mathbf{P}}

\begin{document}

\title{A new version of Toom's proof}
\bookmark[page=1,level=0]{Title}
 \author{Peter G\'acs\thanks{Partially supported by NSF grant CCR-9204284}
 \\ Boston University}
\date{}
\maketitle
\thispagestyle{empty}

  \begin{abstract}
There are several proofs now for the stability of Toom's example of
a two-dimensional stable cellular automaton and its application to
fault-tolerant computation.
Simon and Berman simplified and strengthened Toom's original proof:
the present report is a simplified exposition of their proof.
 \end{abstract}

\section{Introduction}

Let us define cellular automata.

 \begin{definition}
For a finite \( m \), let \( \bbZ_{m} \) be the set of integers modulo \( m \);
we will also write \( \bbZ_{\infty}=\bbZ \) for the set of integers.
A set \( \bbC \) will be called a \df{one-dimensional set of sites}, or
\df{cells}, if it has
the form \( \bbC=\bbZ_{m} \) for a finite or infinite \( m \).
For finite \( m \), and \( x\in\bbC \), 
the values \( x+1 \) \( x-1 \) are always understood modulo \( m \).
Similarly, it will be called a two- or three-dimensional set of sites if it
has the form \( \bbC = \bbZ_{m_1}\times\bbZ_{m_2} \) or 
\( \bbC = \bbZ_{m_1}\times\bbZ_{m_2}\times\bbZ_{m_{3}} \) for finite or
infinite \( m_{i} \).
One- and three-dimensional sets of sites are defined similarly.

For a given set \( \bbC \) of sites and a finite set \( \bbS \) of states,
we call every function \( \xi:\bbC\to \bbS \) a \df{configuration}.
Configuration \( \xi \) assigns state \( \xi(x) \) to site \( x \).
For some interval \( I\subset\rint{0}{\infty} \), a function
\( \eta : \bbC\times I\to\bbS \) will be called a \df{space-time configuration}.
It assigns value \( \eta(x,t) \) to cell \( x \) at time \( t \).

In a space-time vector \( \tup{x,t} \), we will always write the space
coordinate first.
 \end{definition}

 \begin{definition}
Let us be given a function function \( \trans : \bbS^{3}\to \bbS \) and a
one-dimensional set of sites \( \bbC \).
We say that a space-time configuration \( \eta \) 
in one dimension is a \df{trajectory} of the \df{one-dimensional 
(deterministic) cellular automaton} \( \CA(\trans) \)
 \[
   \eta(x,t)=\trans(\eta(x-B,t-T), \eta(x,t-T), \eta(x+B,t-T))
 \]
 holds for all \( x,t \).
Deterministic cellular automata in several dimensions are defined similarly.  
 \end{definition}

Since we want to analyze the effect of noise, 
we will be interested in random space-time configurations.

 \begin{definition}
For a given set \( \bC \) of sites and time interval \( I \), consider a
probability distribution \( \P \) over all space-time configurations
\( \eta:\bbC\times I\to\bbS \).
Once such a distribution is given, we will talk about a 
\df{random space-time configuration} (having this distribution).
We will say that the distribution \( \P \) defines a
\df{trajectory} of the \( \eps \)-\df{perturbation}
 \[
  \CA_\eps(\trans)
 \]
if the following holds.
For all \( x\in\bbC,t\in I \), \( r_{-1},r_0,r_1\in\bbS \), 
let \( E_{0} \) be an event that \( \eta(x+j,t-1)=r_j \) (\( j=-1,0,1 \)) 
and \( \eta(x',t') \) is otherwise fixed in some arbitrary way 
for all \( t'<t \) and for all \( x'\not=x \), \( t'=t \).
Then we have
 \begin{equation*}
   \Pbof{\eta(x,t) = \trans(r_{-1},r_{0},r_{1}) \mid E_{0}} \le\eps.
 \end{equation*} 
 \end{definition}

A simple stable two-dimensional deterministic cellular automaton
given by Toom in \cite{Toom80} can be defined as follows.

 \begin{definition}[Toom rule]
  First we define the neigh\-bor\-hood
 \[
   H = \set{ \tup{0,0}, \tup{0,1}, \tup{1,0} }.
 \]
  The transition function is, for each cell \( x \), a majority vote over
the three values \( x+g_{i} \) where \( g_{i}\in H \).
 \end{definition}

As in ~\cite{GacsReif3dim88}, let us be given an arbitrary
one-dimensional transition function \( \trans \) and the integers \( N,T \).

\begin{definition}
We define the three-dimensional transition function \( \trans' \) as
follows.
The interaction neighborhood is \( H\times\set{-1, 0, 1} \) with the
neighborhood \( H \) defined above.
The rule \( \trans' \) says: in order to obtain your state at time
\( t+1 \), first apply majority voting among self and the northern and
eastern neighbors in each plane defined by fixing the third
coordinate.
Then, apply rule \( \trans \) on each line obtained by fixing the first and
second coordinates.
 
For a finite or infinite \( m \), let \( \bbC \) be our 3-dimensional space
that is the product of \( \bbZ_{m}^{2} \) and a 1-dimensional (finite or
infinite) space \( \bA \) with \( N=|\bA| \).
For a trajectory \( \zg \) of \( \trans \) on \( \bA \), we define the
trajectory \( \zeta' \) of \( \trans' \) on \( \bbC \) by
\( \zeta'(i,j,n,t)=\zg(n,t) \).
\end{definition}

Let \( \zeta' \) be a trajectory of \( \trans' \) and \( \eta \) a trajectory of
\( \CA_\eps(\trans') \) such that \( \eta(w,0)=\zeta'(w,0) \).

  \begin{theorem} \label{t.toomOrig}
Let \( r=24 \), and suppose \( \eps<\frac{1}{32\cdot r^8} \).  
  If \( m=\infty \) then we have
 \[
  \Pbof{ \eta(w,t)\not=\zeta'(w,t) } \le 4 r\eps .
 \]
If \( m \) is finite then we have
 \[
  \Pbof{ \eta(w,t)\not=\zeta'(w,t) }\le 4 r\eps +  (t N)\cdot 2r m^{2}(2r^{2}\eps^{1/12})^{m}.
 \]
 \end{theorem}

The proof we give here is a further simplification of the simplified
proof of \cite{BermSim88}.

 \begin{definition}
Let \( \Noise \) be the set of space-time points \( v \) where \( \eta \) does
not obey the transition rule \( \trans' \).
Let us define a new process \( \xi \) such that \( \xi(w,t)=0 \) if
\( \eta(w,t)=\zeta'(w,t) \), and 1 otherwise.
Let
 \[
 \mathrm{Corr}(a,b,u,t)=
  \Maj(\xi(a,b,u,t),\xi(a+1,b,u,t),\xi(a,b+1,u,t)) .
 \]
 \end{definition}

For all points \( (a,b,u,t+1)\not\in\Noise(\eta) \), we have
 \[
  \xi(a,b,u,t+1) \le
  \max(\mathrm{Corr}(a,b,u-1,t),\mathrm{Corr}(a,b,u,t),
   \mathrm{Corr}(a,b,u+1,t)) .
 \]
Now, Theorem \ref{t.toomOrig} can be restated as follows:

Suppose \( \eps<\frac{1}{32\cdot r^{8}} \).  
If \( m=\infty \) then
 \[
  \Pbof{ \xi(w,t)=1 } \le 4 r\eps .
 \]
 If \( m \) is finite then
 \[
  \Pbof{ \xi(w,t)=1 } \le 4 r\eps + (t N)\cdot 2r m^{2}(2r^{2}\eps^{1/12})^{m}.
 \]

\section{Proof using small explanation trees}

 \begin{definition}[Covering process]
If \( m<\infty \) let \( \bbC'=\bbZ^{3} \) be our \df{covering space}, and
\( \bV'=\bbC'\times \bbZ \) our covering space-time.
There is a projection \( \proj(u) \) from \( \bbC' \) to \( \bbC \) defined by
 \[
  \proj(u)_{i}=u_{i}\bmod m \qquad (i=1,2) .
 \]
This rule can be extended to \( \bbC' \) identically.
We define a random process \( \xi' \) over \( \bbC' \) by
 \[
  \xi'(w,t)=\xi(\proj(w),t) .
 \]
The set \( \Noise \) is extended similarly to \( \Noise' \).
Now, if \( \proj(w_{1})=\proj(w_{2}) \) then \( \xi'(w_{1},t)=\xi'(w_{2},t) \) and
therefore the failures at time \( t \) in \( w_{1} \) and \( w_{2} \) are not
independent.
 \end{definition}

 \begin{definition}[Arrows, forks]
In figures, we generally draw space-time with the time direction
going down.
Therefore, for two neighbor points \( u,u' \) 
of the space \( \bbZ \) (where \( u \) is considered a neighbor for itself as well)  and
integers \( a,b,t \), we will call \df{arrows}, or \df{vertical edges}
the following kinds of (undirected) edges:
 \begin{multline*}
     \set{\tup{a,b,u,t},\tup{a,b,u',t-1 } }, \set{\tup{a,b,u,t}, \tup{a+1,b,u',t-1} } ,
 \\  \set{\tup{a,b,u,t}, \tup{a,b+1,u',t-1} } .
 \end{multline*}
 We will call \df{forks}, or \df{horizontal edges} the following
kinds of edges:
 \begin{multline*}
     \set{\tup{a,b,u,t}, \tup{a+1,b,u,t} }, \set{\tup{a,b,u,t},
       \tup{a,b+1,u,t} } , 
  \\ \set{\tup{a+1,b,u,t}, \tup{a,b+1,u,t} } .
 \end{multline*}
We define the graph \( \bG \) by introducing all possible arrows and
forks.
Thus, a point is adjacent to 6 possible forks and 18 possible
arrows: the degree of \( \bG \) is at most
 \[
  r=24.
 \]
(If the space is \( d+2 \)-dimensional, then \( r=12(d+1) \).)
We use the notation \( \Time(\tup{w,t})=t \).
 \end{definition}

The following lemma is key to the proof, since
it will allow us to estimate the probability of each deviation from the
correct space-time configuration.
It assigns to each deviation a certain tree called its ``explanation''.
Larger explanations contain more noise and have a
correspondingly smaller probability.
For some constants \( c_{1},c_{2} \), there will be \( \le 2^{c_{1}L} \) 
explanations of size \( L \) and each such explanation will have probability
upper bound \( \eps^{c_{2}L} \).

 \begin{lemma}[Explanation Tree]\label{l.explTree}
Let \( u \) be a point outside the set \( Noise' \) with \( \xi'(u)=1 \).
Then there is a tree \( \Expl(u,\xi') \) consisting of \( u \) and points
\( v \) of \( \bG \) with \( \Time(v)<\Time(u) \) and connected with arrows
and forks called an \df{explanation} of \( u \).
It has the property that if \( n \) nodes of \( \Expl \)
belong to \( \Noise' \) then the number of edges of \( \Expl \) is at most
\( 4(n-1) \).
 \end{lemma}

This lemma will be proved in the next section.  
To use it in the proof of the main theorem, we need some easy lemmas.

 \begin{definition}
A \df{weighted tree} is a tree whose nodes have weights 0 or 1, with the
root having weight 0.
The \df{redundancy} of such a tree is the ratio of its number of
edges to its weight.
The set of nodes of weight 1 of a tree \( T \) will be denoted by
\( F(T) \).  

A subtree of a tree is a subgraph that is a tree.
 \end{definition}

 \begin{lemma} \label{l.oneCut}
Let \( T \) be a weighted tree of total weight \( w>3 \) and redundancy \( \lg \).  
It has a subtree of total weight \( w_{1} \) with
 \( w/3 <w_{1} \le 2w/3 \), and redundancy \( \le \lg \).
 \end{lemma}

\begin{proof}
Let us order \( T \) from the root \( r \) down.
Let \( T_{1} \) be a minimal subtree below \( r \) with weight \( > w/3 \).
Then the subtrees immediately below \( T_{1} \) all weigh \( \le w/3 \).   
Let us delete as many of these as possible while keeping \( T_{1} \)
weigh \( >w/3 \).
At this point, the weight \( w_{1} \) of \( T_{1} \) is \( > w/3 \) but \( \le 2w/3 \)
since we could subtract a number \( \le w/3 \) from it so that \( w_{1} \) would
become \( \le w/3 \) (note that since \( w>3 \)) the tree \( T_1 \) is not a
single node.

Now \( T \) has been separated by a node into \( T_{1} \) and \( T_{2} \), with
weights \( w_{1},w_{2}>w/3 \).
Since the root of a tree has weight 0, by definition the possible
weight of the root of \( T_{1} \) stays in \( T_{2} \) and we have \( w_{1}+w_{2}=w \).
The redundancy of \( T \) is then a weighted average of the redundancies
of \( T_{1} \) and \( T_{2} \), and we can choose the one of the two with the
smaller redundancy: its redundancy is smaller than that of \( T \).%
\end{proof}

\begin{theorem}[Tree Separator]
Let \( T \) be a weighted tree with weight \( w \) and redundancy \( \lg \), 
and let \( k<w \).
Then \( T \) has a subtree with weight \( w' \) such that \( k/3< w'\le k \) 
and redundancy \( \le \lg \).
 \end{theorem}

\begin{proof}
Let us perform the operation of Lemma \ref{l.oneCut} repeatedly,
until we get weight \( \le k \).
Then the weight \( w' \) of the resulting tree is \( >k/3 \).%
\end{proof}

\begin{lemma}[Tree Counting]
In a graph of maximum node degree \( r \) the number of weighted
subtrees rooted at a given node and having \( k \) edges is at most
\( 2r\cdot(2r^{2})^k \).
 \end{lemma}

 \begin{proof}
Let us number the nodes of the graph arbitrarily.
Each tree of \( k \) edges can now be traversed in a breadth-first
manner.
At each non-root node of the tree of degree \( i \) from which we
continue, we make a choice out of \( r \) for \( i \) and then a choice out
of \( r-1 \) for each of the \( i-1 \) outgoing edges.
This is \( r^{i} \) possibilities at most.
At the root, the number of outgoing edges is equal to \( i \), so this
is \( r^{i+1} \).
The total number of possibilities is then at most \( r^{2k+1} \) since
the sum of the degrees is \( 2k \).
Each point of the tree can have weight 0 or 1, which multiplies the
expression by \( 2^{k+1} \).%
\end{proof}

\begin{proof}[Proof of Theorem \protect\ref{t.toomOrig}]
Let us consider each explanation tree a weighted tree in which the
weight is 1 in a node exactly if the node is in \( \Noise' \).
For each \( n \), let \( \cE_{n} \) be the set of possible explanation trees
\( \Expl \) for \( u \) with weight \( |F(\Expl)|=n \).
First we prove the theorem for \( m=\infty \), that is \( \Noise'=\Noise \).  
If we fix an explanation tree \( \Expl \) then all the events \( w\in\Noise' \)
for all \( w\in F=F(\Expl) \) are independent from each other.
It follows that the probability of the event \( F\sbs\Noise' \) is at most
\( \eps^{n} \).
Therefore we have
 \[
  \Pbof{\xi(u)=1} \le \sum_{n=1}^{\infty} |\cE_{n}|\eps^{n} .
 \]
By the Explanation Tree Lemma, each tree in \( \cE_{n} \) has at most
\( k=4(n-1) \) edges.
By the Tree Counting Lemma, we have
 \[
 |\cE_{n}|\le 2r\cdot(2r^{2})^{4(n-1)},
 \]
Hence  
 \[
  \Pbof{\xi(u)=1} \le 
2r \eps\sum_{n=0}^{\infty}(16r^{8}\eps)^{n}= 2r\eps(1-16r^{8}\eps)^{-1} .
 \]
If \( \eps \) is small enough to make \( 16r^{8}\eps<1/2 \) then this is \( <4 r\eps \).

In the case \( \bbC\not=\bbC' \) this estimate bounds only the probability
of \( \xi'(u)=1,\ |\Expl(u,\xi')|\le m \), since otherwise the events
\( w\in\Noise' \) are not necessarily independent for \( w\in F \).
Let us estimate the probability that an explanation \( \Expl(u,\xi') \)
has \( m \) or more nodes.
It follows from the Tree Separator Theorem that \( \Expl \) has a
subtree \( T \) with weight \( n' \) where \( m/12\le n'\le m/4 \), and at most
\( m \) nodes.
Since \( T \) is connected, no two of its nodes can have the same
projection.
 Therefore for a fixed tree of this kind, for each node of weight 1
the events that they belong to \( \Noise' \) are independent.
Hence for each tree \( T \) of these sizes, the probability that \( T \) is
such a subtree of \( \Expl \) is at most \( \eps^{m/12} \).
To get the probability that there is such a subtree we multiply by
the number of such subtrees.
An upper bound on the number of places for the root is \( tm^{2}N \).
An upper bound on the number of trees from a given root is obtained
from the Tree Counting Lemma.
Hence
 \[
  \Pbof{ |\Expl(u,\xi')|>m }
   \le 2r t m^{2}N\cdot (2r^{2}\eps^{1/12})^{m} .
 \]
\end{proof}

\section{The existence of small explanation trees}

\subsection{Some geometrical facts}

Let us introduce some geometrical concepts.

 \begin{definition}
  Three linear functionals are defined as follows for \( v=\tup{x,y,z,t} \).
 \[
  L_{1}(v)=-x-t/3 , \qquad L_{2}(v)=-y-t/3 , \qquad L_{3}(v)=x+y+2t/3 .
 \]
 \end{definition}
  Notice \( L_{1}(v)+L_{2}(v)+L_{3}(v)=0 \). 

 \begin{definition}
For a set \( S \), we write 
 \[
   \Size(S)=\sum_{i=1}^{3}\max_{v\in S} L_{i}(v) .
 \]
 \end{definition}
Notice that for a point \( v \) we have \( \Size(\set{v})=0 \).

 \begin{definition}
 A set \( \cS = \set{S_{1},\ldots, S_{n}} \) of sets is \df{connected by
intersection} if the graph \( G(\cS) \) is connected which we obtain by
introducing an edge between \( S_{i} \) and \( S_{j} \) whenever
 \( S_{i}\cap S_{j}\not=\emptyset \).
 \end{definition}

 \begin{definition}
 A \df{spanned set} is an object \( \bbP=\tup{P, v_{1}, v_{2}, v_{3}} \) where \( P \) is
a space-time set and \( v_{i}\in P \).
The points \( v_{i} \) are the \df{poles} of \( \bbP \), and \( P \) is its \df{base set}.
  We define \( \Span(\bbP) \) as \( \sum_{i=1}^3 L_{i}(v_{i}) \).
\end{definition}

\begin{remark}
  As said in the introduction, this paper is an exposition of Toom's more general proof in~\cite{Toom80},
  specialized to the case of the construction in~\cite{GacsReif3dim88}.
  Some of the terminology is taken from~\cite{BermSim88}, and differs from the one in~\cite{Toom80}.
  What is called a spanned set here is called a ``polar'' in~\cite{Toom80}, and its span is called its ``extent'' there.
  In our definition of \( L_{i} \) the terms depending on \( t \) play no role, they just make
  the exposition compatible with~\cite{Toom80}.
\end{remark}

\begin{lemma}[Spanned Set Creation]
 \label{l.createSpan}
 If \( P \) is a set then there is a spanned set \( \tup{P,v_{1},v_{2},v_{3}} \) on \( P \)
with \( \Span(\bbP)=\Size(P) \).
 \end{lemma}
 \begin{proof}
Assign \( v_{i} \) to a point of the set \( P \) in which \( L_{i} \) is maximal.%
\end{proof}

The following lemma is our main tool.

 \begin{lemma}[Spanning]
 \label{l.stokes}
 Let \( \bbL = \tup{L, u_{1}, u_{2}, u_{3}} \) be a spanned set and \( \cM \) be a set of
subsets of \( L \) connected by intersection, whose union covers the poles of \( \bbL \).
Then there is a set \( \set{ \bbM_{1},\dots,\bbM_{n}} \) of spanned sets whose
base sets \( M_{i} \) are elements of \( \cM \), such that the following holds.
Let \( M'_{i} \) be the set of poles of \( \bbM_{i} \).
  \begin{cjenum}
   \item \( \Span(\bbL) = \sum_{i}\Span(\bbM_{i}) \).
   \item The union of the sets \( M'_{j} \) covers the set of poles of \( \bbL \).
   \item The system \( \set{M'_{1},\ldots,M'_{n} } \) is a minimal system
connected by intersection (that is none of them can be deleted) 
that connects the poles of \( \bbL \).
  \end{cjenum}
 \end{lemma}
 \begin{proof}
Let \( M_{i_{j}}\in\cM \) be a set containing the point \( u_{j} \).
Let us choose \( u_{j} \) as the \( j \)-th pole of \( M_{i_{j}} \).
Now leave only those sets of \( \cM \) that are needed for a 
minimal tree \( \cT \) of the graph \( G(\cM) \)
connecting \( M_{i_{1}},M_{i_{2}},M_{i_{3}} \).
Keep deleting points from each set (except \( u_{j} \) from \( M_{i_{j}} \)) until
every remaining point is necessary for a connection among \( u_{j} \).
There will only be two- and three-element sets, and any two of them
intersect in at most one element.
Let us draw an edge between each pair of points if they belong to a
common set \( M'_{i} \).
This turns the union
 \begin{align*}
  V = \bigcup_{i} M'_{i}
 \end{align*}
into a graph.
(Actually, this graph can have only two simple forms: a point
connected via disjoint paths to the poles 
\( u_{i} \) or a triangle connected via disjoint paths to these poles.)
For each \( i \) and \( j \), there is a shortest path between \( M'_{i} \) and \( u_{j} \).
The point of \( M'_{i} \) where this path leaves \( M'_{i} \) will be made the
\( j \)-th pole \( u_{ij} \) of \( M_{i} \).
For \( j\in\set{1,2,3} \) we have \( u_{i_{j}j}=u_{j} \) by definition.
This rule creates three poles in each \( M_{i} \) and each point of \( M'_{i} \) is a
pole.

Let us show \( \sum_{i}\Span(\bbM_{i})=\Span(\bbL)  \).
We can write
 \begin{align}\label{eq:sum}
   \sum_{i}\Span(\bbM_{i}) =\sum_{v\in V}\sum_{i,j : v = u_{ij} }L_{j}(v). 
 \end{align}
For a point \( v\in V \), let 
 \begin{align*}
  I(v) =\setof{i : v\in M'_{i}}.
 \end{align*}
For \( i\in I(v) \) let \( E_{i}(v) \) be the set of those \( j\in\set{1,2,3} \) for which
either \( i=i_{j} \) or \( v\ne u_{ij} \).
Because graph \( \cT \) is a tree, for each fixed \( v \) the sets \( E_{i}(v) \) are disjoint.
Because of connectedness, they form a partition of the set \( \set{1,2,3} \).
Let \( e_{i}(j,v)=1 \) if \( j\in E_{i}(v) \) and 0 otherwise, then we have
\( \sum_{i} e_{i}(j,v)=1 \) for each \( j \).

We can now rewrite the sum~\eqref{eq:sum} as
 \begin{align*}
 \sum_{j=1}^{3}\sum_{v\in V}L_{j}(v)(e_{i_{j}}(j,v) +
       \sum_{v\in V}\sum_{i\in I(v)\xcpt\{i_{j}\}} (1-e_{i}(j,v))).
 \end{align*}
If \( i=i_{j}\in I(v) \) then by definition we have \( 1-e_{i}(j,v)=0 \), 
therefore we can simplify the sum as
 \begin{align*}
 \sum_{j=1}^{3}\sum_{v\in V}L_{j}(v)e_{i_{j}}(j,v) + 
\sum_{i\in I(v)}\sum_{j=1}^{3}L_{j}(v)(1-e_{i}(j,v)).
 \end{align*}
The first term is equal to \( \Span(\bbL) \); we show that the last term is 0.
Moreover, we show \( 0=\sum_{j=1}^{3}L_{j}(v)\sum_{i\in I(v)}(1-e_{i}(j,v)) \)
for each \( v \).
Indeed, \( \sum_{i\in I(v)}(1-e_{i}(j,v)) \) is independent of \( j \) since
it is \( |I(v)| - \sum_{i}e_{i}(j,v)=|I(v)|-1 \).
On the other hand, \( \sum_{j}^{3}L_{j}(v)=0 \) as always.
\end{proof}

\subsection{Building an explanation tree}

Let us define the excuse of a space-time point.
 
 \begin{definition}[Excuse]
Let \( v=\tup{a,b,u,t+1} \) with \( \xi'(v)=1 \).
If \( v\not\in\Noise' \) then there is a \( u' \) such that \( \xi'(w)=1 \) for
at least two members \( w \) of the set
 \[
  \Setof{ \tup{a,b,u',t}, \tup{a+1,b,u',t}, \tup{a,b+1,u',t} } .
 \]
We define the set \( \Excuse(v) \) as such a pair of elements \( w \),
and as the empty set in all other cases.
By Lemma \ref{l.createSpan}, we can turn \( \Excuse(v) \) into a spanned
set, \( \tup{\Excuse(v),w_{1},w_{2},w_{3}} \) with span 1.
Denote 
 \begin{align*}
 \Excuse_{i}(v)=w_{i}.
 \end{align*}
Since no excuse is built from a node in \( \Noise' \), let us delete all arrows leading down
from nodes in \( \Noise' \): the new graph is denoted by \( \bG' \).
 \end{definition}

The following lemma utilizes the fact that Toom's rule ``makes triangles
shrink''.

\begin{lemma}[Excuse size]\label{l.spanIncr}
If \( \bbV = \tup{V, v_{1},v_{2},v_{3}} \) is a spanned set and \( v_{i} \) are not in
\( \Noise' \) then we have
 \[
 \sum_{j=1}^{3}L_{j}(\Excuse_{j}(v_{j}))=\Span(\bbV)+1 .  
 \]
 \end{lemma}

 \begin{proof}
Let \( T \) be the triangular prism 
 \begin{align*}
 (0,0,0,-1)+\setof{u : L_{1}(u)\le 0,\, L_{2}(u)\le 0,\, L_{3}(u)\le 1}.
 \end{align*}
We have \( \Size(T)=1 \), and \( \Excuse(v) \sbs v+T \).
Since the chosen poles turn \( \Excuse(v) \) into a spanned set of size 1,
the function \( L_{j} \) achieves its maximum in \( v+T \) on \( \Excuse_{j}(v) \).
We have
  \[
  L_{j}(\Excuse_{j}(v))=\max_{u\in v+T} L_{j}(u) = L_{j}(v) + \max_{u\in T} L_{j}(u).
  \]
  Hence we have
 \begin{align*}
  \sum_{j}L_{j}(\Excuse_{j}(v_{j})) &=
  \sum_{j}\max_{u\in T}L_{j}(u)+\sum_{j} L_{j}(v_{j})
 \\                         &= \Size(T)+\Span(\bbV)=1+\Span(\bbV).
 \end{align*}
\end{proof}

 \begin{definition}[Clusters]
Let us call two nodes \( u,v \) of the above graph with \( \Time(u)=\Time(v)=t \)
\df{equivalent} if there is a path between them in \( \bG' \) made of arrows,
using only points \( x \) with \( \Time(x)\le t \).
An equivalence class will be called a \df{cluster}.
For a cluster \( K \) we will denote by \( \Time(K) \) the time of its points.
We will say that a fork or arrow \df{connects} two clusters if it connects
some of their nodes.
 \end{definition}

By our definition of \( \bG' \), if a cluster contains a point in \( \Noise' \) then it contains no other points.

 \begin{definition}[Cause graph]
Within a subgraph \( \bG' \) that is some excuse graph,
for a cluster \( K \), we define the \df{cause graph} 
\( G_{K}=\tup{V_{K},E_{K}} \) as follows.
The elements of \( G_{K} \) are those clusters \( R \) with
\( \Time(R)=\Time(K)-1 \) which are reachable by an arrow from \( K \).
For \( R,S\in V_{K} \) we have \( \{R,S\}\in E_{K} \) iff for some
\( v\in R \) and \( w\in S \) we have \( \Time(v)=\Time(w)=\Time(K)-1 \) and
\( \{v,w\}\in\Forks \).
 \end{definition}

 \begin{lemma} \label{l.connSlice} 
The cause graph \( G_{K} \) is connected.
 \end{lemma}

 \begin{proof}
The points of \( K \) are connected via arrows using
points \( x \) with \( \Time(x) \le \Time(K) \).
The clusters in \( G_{K} \) are therefore connected with each other only through
pairs of arrows going trough \( K \).
The tails of each such pair of arrows in \( \Time(K)-1 \) are connected by a fork.%
\end{proof}

 \begin{definition}
A \df{spanned cluster} is a spanned set that is a cluster.
 \end{definition}

The explanation tree will be built from an intermediate object defined
below.
Let us fix a point \( u_{0} \): from now on we will work in the subgraph of the
graph \( \bG' \) reachable from \( u_{0} \) by arrows pointing backward in time.
Clusters are defined in this graph.

 \begin {figure}[ht]
\begin {equation*}
\includegraphics{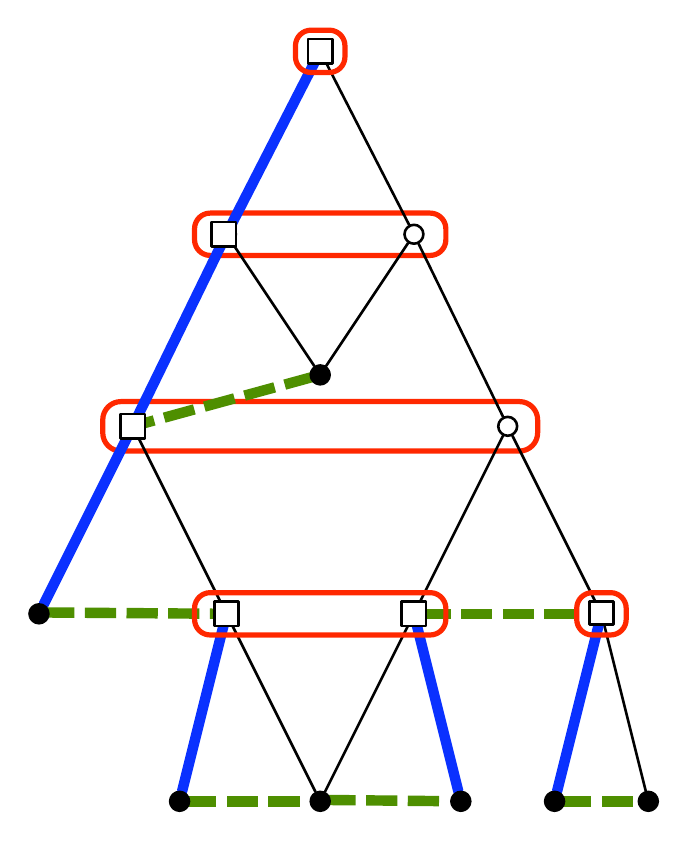}
\end {equation*}
\caption {An explanation tree.  
The black points are noise.
The squares are other points of the explanation tree.
Thin lines are arrows not in the explanation tree.
Framed sets are clusters to which
the refinement operation was applied.
Thick solid lines are arrows, thick broken lines are forks of the
explanation tree.
\label{f.expl-tree}}
 \end {figure}

 \begin{definition}
A \df{partial explanation tree} is an object of the form
\( \tup{C_{0},C_{1},E} \).
Elements of \( C_{0} \) are spanned clusters called \df{unprocessed nodes}, 
elements of \( C_{1} \) are \df{processed nodes}, these are nodes of \( \bG \).
The set \( E \) is a set of arrows or forks
between processed nodes, between poles of the spanned clusters, and between
processed nodes and poles of the spanned clusters.
From this structure a graph is formed if we identify each pole of a spanned
cluster with the cluster itself.
This graph is required to be a tree.

The \df{span} of such a tree will be the sum \( \Span(T) \) of the spans
of its unprocessed clusters and the number of its forks.
 \end{definition}

The explanation tree will be built by applying repeatedly a
``refinement'' operation to partial explanation trees.

 \begin{definition}[Refinement]
Let \( T \) be a partial explanation tree, and let the spanned cluster
\( \bbK=\tup{K,v_{1},v_{2},v_{3}} \) be one of its unprocessed nodes, with
\( v_{i} \) not in \( \Noise' \).
We apply an operation whose result will be a new tree \( T' \).

Consider the cause graph \( G_{K}=\tup{V_{K},E_{K}} \) defined above.
Let \( \cM=V_{K}\cup E_{K} \), that is the family of all clusters in \( V_{K} \) 
(sets of points) and all edges in \( G_{K} \) connecting them, (two-element
sets).
Let \( L \) be the union of these sets, and \( \bbL=\tup{L, u_{1},u_{2},u_{3}} \)
a spanned set where \( u_{i}=\Excuse_{i}(v_{i}) \).
Lemma~\ref{l.connSlice} implies that the set \( \cM \) is connected by
intersection.
Applying the Spanning Lemma~\ref{l.stokes} to \( \bbL \) and \( \cM \),
we find a family \( \bbM_{1},\dots,\bbM_{n} \) of spanned sets with
 \[ 
  \sum_{i}\Span(\bbM_{i}) = \Span(\bbL)=\sum_{i}L_{i}(u_{i}).
 \]
It follows from Lemma~\ref{l.spanIncr} that the latter sum is
\( \Span(\bbK)+1 \), and that \( u_{i} \) are among the poles of these sets.
Some of these sets are spanned clusters, others are forks
connecting them, adjacent to their poles.
Consider these forks again as edges and the spanned clusters as
nodes.
By the minimality property of Lemma~\ref{l.stokes}, they form a tree
\( U(\bbK) \) that connect the three poles of \( \bbL \).

The refinement operation takes an unprocessed node
\( \bbK=\tup{K,v_{1},v_{2},v_{3}} \) in the tree \( T \).
This node is connected to other parts of the tree by some of its poles
\( v_{j} \).

The  operation deletes cluster \( K \), and keeps those poles \( v_{j} \)
that were needed to keep connected \( \bbK \) to other clusters and nodes in \( T \).
It turns these into processed nodes, and adds the tree \( U(\bbK) \) just built,
declaring each of its spanned clusters unprocessed nodes.
Then it adds the arrow from these \( v_{j} \) to \( \Excuse_{j}(v_{j}) \).
Even if none of these nodes were needed for connection, it keeps \( v_{1} \)
and adds the arrow from \( v_{1} \) to \( \Excuse_{1}(v_{1}) \).
 \end{definition}

The refinement operation increases both the span and the
number of arrows by 1.

Let us build now the explanation tree.
We start with a node \( u_{0}\not\in\Noise' \) with \( \xi'(u_{0})=1 \) and from
now on 
work in the subgraph of the graph \( \bG \) of points reachable from \( u_{0} \) by
arrows backward in time.
Then \( \tup{\{u_{0}\},u_{0},u_{0},u_{0}} \) is a spanned cluster, forming a
one-node partial explanation tree if we declare it an unprocessed node.
We apply the refinement operation to this partial explanation
tree, as long as we can. 
When it cannot be applied any longer then all nodes are either processed
or one-point spanned clusters belonging to \( \Noise' \).
See an example in Figure~\ref{f.expl-tree}.

\begin{proof}[Proof of Lemma \protect\ref{l.explTree}]
What is left to prove is the estimate on the number of edges of our explanation tree \( T \).
Note the following:
\begin{itemize}
\item The span of \( T \) is the number of its forks.
\item Each point at some time \( t \) that is not in \( \Noise' \) is incident to some arrows going to time \( t-1 \).
\end{itemize}
Let us contract each arrow \( \tup{u,v} \) of \( T \) one-by-one into its bottom point \( v \).
The edges of the resulting tree \( T' \) are the forks.
All the processed nodes will be contracted into the remaining
one-node clusters that are elements of \( \Noise' \).
If \( n \) is the number of these nodes then there are \( n-1=\Span(T) \)
forks in \( T' \).

The number of arrows in \( T \) is at most \( 3(n-1) \).
Indeed, each introduction of at most 3 arrows by the refinement operation
was accompanied by an increase of the span by 1.
The total number of edges of \( T \) is thus at most \( 4(n-1) \).
\end{proof}

 \bibliographystyle{amsplain}
 \bibliography{reli,gacs-publ}






\end{document}